\documentclass[11pt]{article}
\usepackage{latexsym}
\usepackage{theorem}
\usepackage{graphicx}
\usepackage{amsmath,color}
\usepackage{amsfonts}
\usepackage{natbib}
\usepackage{soul}
\usepackage{enumerate}
\usepackage{hyperref}
\usepackage{tikz}
\usepackage{caption}
\usepackage{subcaption}

\theorembodyfont{\rmfamily}
\newtheorem{theorem}{Theorem}
\newtheorem{conjecture}[theorem]{Conjecture}
\newtheorem{lemma}[theorem]{Lemma}
\newtheorem{proposition}[theorem]{Proposition}
\newtheorem{corollary}[theorem]{Corollary}
\newtheorem{definition}[theorem]{Definition}

\theoremstyle{break}

\newtheorem{example}[theorem]{Example}

\usepackage{mathtools}

\newenvironment{proof}{\paragraph{Proof.}}{\hfill$\square$}

\title{Optimal Pure Strategies for a Discrete Search Game}

\date{}
\author{Thuy Bui\thanks{Rutgers Business School, 1 Washington Park, Newark, NJ 07102, USA, tb680@business.rutgers.edu} \and Thomas Lidbetter\thanks{Rutgers Business School, 1 Washington Park, Newark, NJ 07102, USA, tlidbetter@business.rutgers.edu}\thanks{Department of Systems and Information Engineering, University of Virginia, VA 22903, USA, tlidbetter@virginia.edu} 
\and Kyle Y. Lin \thanks{Operations Research Department, Naval Postgraduate School, Monterey, CA 93943}}

\providecommand{\keywords}[1]{\textbf{\textbf{Keywords:}} #1}

\begin{document}
	
\maketitle

\begin{abstract}
\noindent
Consider a two-person zero-sum search game between a Hider and a Searcher.
The Hider chooses to hide in one of $n$ discrete locations (or ``boxes'') and the Searcher chooses a search sequence specifying which order to look in these boxes until finding the Hider.
A search at box $i$ takes $t_i$ time units and finds the Hider---if hidden there---independently with probability $q_i$, for $i=1,\ldots,n$.
The Searcher wants to minimize the expected total time needed to find the Hider, while the Hider wants to maximize it.
It is shown in the literature that the Searcher has an optimal search strategy that mixes up to $n$ distinct search sequences with appropriate probabilities.
This paper investigates the existence of optimal pure strategies for the Searcher---a single deterministic search sequence that achieves the optimal expected total search time regardless of where the Hider hides.
We identify several cases in which the Searcher has an optimal pure strategy, and several cases in which such optimal pure strategy does not exist.
An optimal pure search strategy has significant practical value because the Searcher does not need to randomize their actions and will avoid second guessing themselves if the chosen search sequence from an optimal mixed strategy does not turn out well.
\end{abstract}




\newpage
\section{Introduction}
A major criticism of game theory is the notion of randomized (or {\em mixed}) strategies: see, for example, \cite{binmore} and \cite{gintis}. In a Nash equilibrium, players randomize between strategies in a precisely specified way, yet they are indifferent between each of the strategies between which they randomize, given the behavior of the other players. Intuitively, it is difficult to explain such precision in the players' random actions. \cite{harsanyi} had an elegant solution to this criticism, by defining a Bayesian game where the payoffs of a given game are slightly perturbed by a random amount. Any equilibrium in the original game can ``almost always'' be obtained as a limit of a pure strategy equilibria in the Bayesian game, as the perturbations go to zero.

This is all very well, but in practice, if a Nash equilibrium requires a player to randomize, there are also issues about the practical implementation of such randomized strategies. This is particularly salient in search games, which can be used to model military operations where lives may be at stake and randomized actions can be hard to justify to decision makers.  If there is a choice between a mixed strategy solution and a pure strategy (non-randomized) solution, we argue that the latter would be preferable, since it is more easily implementable and explainable. Decision makers do not need to rely on the roll of a die or outcome of a roulette wheel to make their choices.

In this paper we consider the problem of finding optimal pure strategies for a generalization of a classical search game introduced by \cite{bram}, where a Hider chooses one of a finite number $n$ of locations (or ``boxes'') in which to hide and a Searcher inspects the boxes one-by-one until finding the Hider. When a given box is searched, if the Hider is in that box, there is a specified probability (the ``detection probability'') that the Hider is found. The Searcher aims to minimize the expected number of boxes to search before finding the Hider, whereas the Hider aims to maximize it. The game is zero-sum, and \cite{bram} proved that an optimal (min-max) mixed search strategy exists and suggested a numerical way to compute it. There has been sporadic interest in the game, but not much is known in terms of closed form solutions. Indeed, it seems unlikely that such solutions exist in general. \cite{RG78} found an optimal Hider strategy in the case of two boxes under certain, rather specific, conditions and \cite{Ruckle} found an optimal search strategy for the case that the detection probabilities are the same for all boxes.

There has been some recent interest in a more general version of the game described above, where each box takes a certain amount of time to search. The payoff becomes the total expected time to find the Hider. 
\cite{CLG22} showed that optimal strategies exist and that there exists an optimal Searcher strategy that mixes between at most $n$ pure strategies.
\cite{CL23} presented an algorithm to numerically calculate optimal strategies for the players.

The question of whether there exist optimal {\em pure} strategies in this game has received little attention. \cite{Ruckle} pointed out that for the case of equal search times and detection probabilities all equal to $1/2$, there is an optimal pure strategy, but prior to this paper, there were no other known cases for which there is an optimal pure strategy. We rectify this by showing that there are optimal pure strategies for several classes of the game, and we show how to calculate them. We also give conditions under which there is no optimal pure strategy. Roughly speaking, this is the case when the detection probabilities are high. In this case, there is a high probability of discovery after the first box has been searched, which immediately skews the expected search time (in the negative direction) towards this box.

This paper is laid out as follows. In Section~\ref{sec:prelim}, we recall the definition of the game and some previous results. Although this paper is primarily on the topic of optimal pure strategies, we also give a full solution to the game in mixed strategies for the case of equal detection probabilities in Section~\ref{sec:equal}. This case was previously unsolved. In Section~\ref{sec:opt-pure}, we use a series of lemmas to prove our main result, Theorem~\ref{lem:m_seq_diff_prob}, which gives conditions under which a pure strategy solution to the game exists. We then apply Theorem~\ref{lem:m_seq_diff_prob} to several classes of the game. We end Section~\ref{sec:opt-pure} by proving conditions under which no pure strategy solution exists, and we explore the consequences of this result. In Section~\ref{sec:two}, we restrict attention to the case of two boxes, taking advantage of the simpler structure to strengthen our results about the existence of pure strategy solutions. We go on to consider the question of what is a ``best'' pure strategy in the case that no optimal pure strategy exists.  In Section~\ref{sec:conclusion} we conclude.

\section{Preliminaries} \label{sec:prelim}

In this section, we recall the definition of the game $G$ and state some previous results that we will use later on in the paper.

\subsection{Definition of the Search Game} \label{sec:def}

The game $G$ is played between a Hider (the maximizer) and a Searcher (the minimizer). The Hider's pure strategies are a set of $n$ boxes, denoted $[n]\equiv \{1,\ldots,n\}$, in which he may hide. Hence, the set of mixed strategies for the Hider is 
\[
\Delta^n \equiv \{(p_1,\ldots,p_n): p_i \ge 0 \text{ for } i=1,\ldots,n \text{ and } \sum_{i=1}^n p_i = 1 \}.
\]
Here, $p_i$ is the probability the Hider chooses to hide in box $i$.

A pure strategy for the Searcher is a search sequence---an infinite sequence of the boxes corresponding to the order in which the Searcher searches them.
Thus, the Searcher's pure strategy set is the set $\mathcal{C} \equiv [n]^\infty$, an infinite set.
As in \cite{CLG22}, we define a mixed strategy for the Searcher as a probabilistic choice of a {\em countable} subset of pure strategies. More precisely, it is a function $\theta:\mathcal{C} \rightarrow [0,1]$ such that $\{\xi \in \mathcal{C}: \theta(\xi)>0 \}$ is countable and 
\[
\sum_{\xi \in \mathcal{C}} \theta(\xi) = 1.
\]
The condition that $\{\xi \in \mathcal{C}: \theta(\xi)>0 \}$ is countable ensures that this sum is well-defined.

Each box $i$ has a {\em search time} $t_i >0$, which is the time required to complete one search of the box. Box $i$ also has a {\em detection probability} $q_i \in (0,1]$, which is the probability the Searcher will find the Hider after one search of box $i$ if the Hider is hidden there.
Write $r_i = 1-q_i$ for the overlook probability of box $i$.
The outcome of each search is independent of all previous search outcomes.
For each pure strategy pair $(i, \xi)$---when the Hider hides in box $i$ and the Searcher uses the search sequence $\xi$--- we write $u(i,\xi)$ for the payoff of the game, which is the expected time needed for the Searcher to find the Hider under the pure strategy pair $(i, \xi)$. 
We refer to the payoff as the {\em expected search time}. When the players use mixed strategies $\mathbf p$ and $\theta$, we extend $u$ to denote the expected payoff as follows.
\[
u(\mathbf p, \theta) \equiv \sum_{i=1}^n \sum_{\xi \in \mathcal{C}} p_i \theta(\xi) u(i,\xi).
\]
We similarly extend the definition of $u$ for cases when one player uses a mixed strategy and the other uses a pure strategy.

For a given Hider mixed strategy $\mathbf{p}$, the problem of finding a best response---that is, a Searcher strategy $\xi$ that minimizes the expected search time $\sum_{i=1}^n p_i u(i, \xi)$ against $\mathbf p$---is well understood. The solution, first discovered by Blackwell \cite[reported in][]{Matula}, can be found recursively. The Searcher should first search any box $i$ that maximizes the probability of detection per unit time, namely $p_i q_i/t_i$. After each subsequent search, the next box to be searched is found by updating the hiding probabilities $p_i$ according to Bayes' law and repeating the calculation. It is straightforward to show that after box $i$ has been searched $m_i$ times for $i=1,\ldots,n$, the next box to be searched should be some $i$ that maximizes the index 
\begin{align}
\psi_i \equiv \frac{p_i q_i r_i^{m_i}}{t_i}. 
\label{eq:index}
\end{align}
In order to study optimal pure strategies, in this paper, we assume there are some positive coprime integers $k_1,\ldots,k_n$ and some $r >0$ such that $r_i^{k_i}=r$ for all $i\in[n]$.
It is worth noting that this assumption is not at all restrictive because if it were not the case, we could find $k_1,\ldots,k_n$ and $r$ such that $r_i^{k_i}$ is arbitrarily close to $r$ for each $i$.
With $r_i^{k_i}=r$ for all $i\in[n]$, the index in \eqref{eq:index} becomes
\begin{align}\label{eq:psi}
    \psi_i \equiv \frac{p_i q_i r^{m_i/k_i}}{t_i}.
\end{align}

Following the terminology of \cite{CLG22}, we refer to any sequence that can be produced in this way (and is hence a best response to some Hider strategy $\mathbf p$) as a {\em Gittins search sequence} (against $\mathbf p$). The name comes from a comment by Kelly in \cite{Gittins79} noting that Blackwell?s solution is equivalent to a Gittins index policy  obtained by modeling the search as a tractable version of the multiarmed bandit problem \cite[see][]{Gittins11}.

There may be multiple Gittins search sequences against a given $\mathbf p$ in the case that there are ties between the indices during the search process.  Any permutation $\sigma$ of $[n]$ can be used to define a tie-breaking rule, where each tie is broken by choosing the box  that appears first in $\sigma$. We call a Gittins search sequence {\em consistent} if all ties are broken using the same tie-breaking rule given by some permutation $\sigma$ of $[n]$. There are at most $n!$ consistent Gittins search sequences against any given $\mathbf p$, and two different permutations may result in the same consistent Gittins search sequence.

A particular Hider strategy of interest is the {\em equalizing strategy}, defined as the strategy that makes the Searcher indifferent between all the boxes at the beginning of the search.
We define the Hider's equalizing strategy below.
\begin{definition}
The Hider's equalizing strategy $\mathbf{p}^*$ is given by
\begin{equation*}
p^*_i = \frac{t_i/q_i}{\sum_{j=1}^n t_j/q_j}, \qquad i\in[n].
\label{eq:equalizing}
\end{equation*}
\end{definition}

\bigskip

The Hider's equalizing strategy $\mathbf{p}^*$ appears to be a strong strategy because it makes all boxes equally attractive (or unattractive) to the Searcher at the beginning of the search.
In particular, there are $n!$ distinct consistent Gittins search sequences against $\mathbf{p}^*$. 
\cite{RG78} showed that the Hider's equalizing strategy is indeed optimal in certain cases.
We recite the theorem below, which will be used later.
\begin{theorem}[Theorem 2 of \cite{RG78}] \label{thm:RG78}
Consider the game $G$ with $n=2$ and $t_1=t_2$.
If $r_1^k=r_2^{k+1}$ for some positive integer $k$, where $r_i \equiv 1-q_i$ for $i=1,2$, then the Hider's equalizing strategy $\mathbf{p}^*$ is optimal if and only if $k \le 12$.
\end{theorem}

If the Hider uses the equalizing strategy, then the indices $p^*_iq_i/t_i$ are all equal to some $\lambda$, and the index $\psi_i$ from~(\ref{eq:psi}) reduces to $\psi_i=\lambda r^{m_i/k_i}$. In this case, it is easy to see that any Gittins search sequence against $\mathbf p^*$ will be of the form $s_1,s_2,\ldots$, where each $s_j, j=1,2,\ldots$ is a subsequence containing box $i$ exactly $k_i$ times for each $i \in [n]$. Any {\em consistent} Gittins search sequence against $\mathbf p^*$ will be an infinite repetition of one such subsequence $s_j$. We denote such a strategy $(s_j)$. The length of time each subsequence $s_j$ spends searching is equal to $\sum_{i \in [n]} k_i t_i$, and we denote this quantity by $T$.

\subsection{The Value of the Game and Optimal Strategies}

As discussed in \cite{CLG22}, it follows from standard results on semi-infinite games that the game $G$ has a value.  That is, there exists some $V$ such that
\[
V = \max_{\mathbf{p}\in \Delta^n} \inf_{\xi \in \mathcal{C}}  u(\mathbf p, \xi) = \inf_{\theta} \max_{i \in [n]} u(i,\theta),
\]
where the infinum on the right-hand side is taken over all Searcher mixed strategies $\theta$. 

It also follows from standard results that there exists an optimal Hider strategy in $G$ (that is, a mixed strategy for the Hider that guarantee a payoff of at least $V$) and an $\varepsilon$-optimal Searcher strategies (that is, mixed strategies for the Searcher that guarantee a payoff of at most $V+\varepsilon$, for any $\varepsilon>0$).

A major result from \cite{CLG22}, which we will use later in this paper, is a stronger statement on the existence of optimal Searcher strategies. 

\begin{theorem}[Theorem 3 of \cite{CLG22}] \label{thm:CLG}
There exists an optimal mixed strategy for the Searcher in $G$ that mixes between at most $n$ pure strategies that are all consistent Gittins search sequences against the optimal Hider strategy.
\end{theorem}




\section{The Case of Equal Detection Probabilities} \label{sec:equal}
In this section, we consider the game with $q_i=q$ for all $i \in [n]$, for some $q\in (0,1)$, but arbitrary search times $t_1, \ldots, t_n$.
In the special case where $t_1=\cdots=t_n$, \cite{Ruckle} found that the Hider's equalizing strategy $\mathbf{p}^* = (1/n, 1/n, \ldots, 1/n)$ is optimal.
For the Searcher, it is optimal to choose with equal probability each of the search sequences $(i,i+1,\ldots,n,1,2,\ldots,i-1)$, for $i=1,\ldots,n$.

In another special case with $q=1$, each box need only be searched once to ensure that the Hider is found with probability $1$. In this case, the Hider's equalizing strategy $\mathbf{p}^*$ chooses each box $i$ with probability proportional to $t_i$, and any permutation of $[n]$ is a Gittins search sequence against $\mathbf{p}^*$. The parameter $T$, defined at the end of Subsection~\ref{sec:def} reduces to $T=\sum_{i \in [n]} t_i$.
This special case was first considered in \cite{Condon09}, and independently in \cite{AL13} and \cite{Lidbetter13}. The optimal strategies found for the Searcher were different in each of these works, but the optimal Hider strategy is unique \cite[see][]{AL13}. We summarize the solution below, giving the optimal Searcher strategy from \cite{Condon09}, which is arguably the most elegant of the three. 

\begin{theorem} \label{thm:q=1}
In the game $G$ with $q_i=1$ for all $i \in [n]$, it is optimal for the Hider to hide in box $i$ with probability proportional to $t_i$. It is optimal for the Searcher to choose the first box $i$ she looks in with probability proportional to $t_i$, then to search the remaining boxes in the order $i+1,i+2,\ldots,n,1,2,\ldots,i-1$. The value $V$ of the game is given by
\[
V =  T - \frac{1}{T}\sum_{1\le i < j \le n} t_i t_j.
\]
\end{theorem}

Our contribution in this paper is to extend Theorem~\ref{thm:q=1} to the case in which $q < 1$ in the next theorem.

\begin{theorem} \label{thm:q-equal}
Denote by $\theta(\sigma)$ the probability of choosing permutation $\sigma$ of $[n]$ for some arbitrary optimal Searcher strategy for the game $G$ with $q=1$.
For the game with equal detection probabilities $q \in (0,1)$, the Searcher strategy that chooses each search sequence $(\sigma)$---which repeats $\sigma$ indefinitely---with probability $\theta(\sigma)$ is optimal for the Searcher. The Hider's equalizing strategy $\mathbf{p}^*$---choosing box $i$ with probability $p^*_i=t_i/T$---is optimal.  The value of the game is 
\[
V = \frac{T}{q} - \frac{1}{T}\sum_{1\le i < j \le n } t_i t_j.
\]
\end{theorem}
\begin{proof}
First suppose the Hider uses the strategy $\mathbf p^*$, and let $\xi$ be an arbitrary Searcher pure strategy which is a best response to $\mathbf p^*$. It must be a Gittins search sequence against $\mathbf p^*$, so it must be of the form $\sigma_1,\sigma_2,\ldots$, where each $\sigma_j, j\ge 1$ is a permutation of $[n]$. For each $k=1,2,\ldots$, we refer to the {\em $k$th round of the search} as boxes number $(k-1)n+1,\ldots,kn$ to be searched. The probability the Hider is found in the $k$th round of search is $(1-q)^{k-1}q$. Also, by Theorem~\ref{thm:q=1}, the conditional expected search time given the Hider is found in the $k$th round of search is 
\[
V_k \equiv  (k-1)T + \left(T - \frac{1}{T}\sum_{1\le i < j \le n } t_i t_j \right) = kT -  \frac{1}{T}\sum_{1\le i < j \le n } t_i t_j 
\]
Hence, the expected search time of $\xi$ against $\mathbf p^*$ is
\begin{align*}
 u(\mathbf p^*, \xi)&= \sum_{k=1}^\infty (1-q)^{k-1}q V_k \\
& = \sum_{k=1}^\infty (1-q)^{k-1}q \left( kT  - \frac{1}{T}\sum_{1\le i < j \le n } t_i t_j \right)\\
&= \frac{T}{q} - \frac{1}{T}\sum_{1\le i < j \le n } t_i t_j.
\end{align*}
Thus, the value satisfies $V \ge T/q - (1/T) \sum_{1\le i < j \le n } t_i t_j$.

We now turn to the Searcher strategy given by $\theta$. For any fixed Hider strategy, the probability the Hider is found in the $k$th round of search is $(1-q)^{k-1}q$ and by Theorem~\ref{thm:q=1}, the conditional expected search time given the Hider is found in the $k$th round of search is $V_k$, so by the same calculation, the expected search time is $T/q - (1/T) \sum_{1\le i < j \le n } t_i t_j$. This provides an upper bound on $V$, and we have equality.
\end{proof}

Theorem~\ref{thm:q-equal} generalizes a result of \cite{Ruckle} that solves the special case of equal detection probabilities and unit search times. In this case, as \cite{Ruckle} found, it is optimal for the Searcher to choose with equal probability each of the Searcher strategies $(i,i+1,\ldots,n,1,2,\ldots,i-1)$, for $i=1,\ldots,n$. 
Our next result shows that, for this special case, the Searcher has an optimal strategy that mixes between only two pure strategies.

\begin{proposition} \label{prop:equal}
    In the game $G$ with $q_i=q$ for all $i \in [n]$ and $t_i=t$ for all $j\in [n]$, it is optimal for the Searcher to choose with equal probability each of the search sequences $(1,2,\ldots,n)$ and $(n,n-1,\ldots,1)$.
\end{proposition}
\begin{proof}
By Theorem~\ref{thm:q-equal}, it is sufficient to show that in the case $q=1$, it is optimal to choose with equal probability each of the sequences $1,2,\ldots,n$ and $n,n-1,\ldots,1$. Indeed for any box $i$, this strategy $\theta$ guarantees an expected search time of 
\[
u(i,\theta) = \frac{1}{2}(it) + \frac{1}{2}(n-i+1)t = \left(\frac{n+1}{2}\right)t.
\]
But by Theorem~\ref{thm:q=1}, the value of the game in this case is
\[
V=nt-\frac{1}{nt}{n \choose 2}t^2 = \left(\frac{n+1}{2}\right)t.
\]
Hence, $\theta$ is optimal.
\end{proof}

\section{Optimal Pure Strategy Solutions} \label{sec:opt-pure}
Because the optimal Hider strategy must be a mixed strategy, the study of optimal pure Searcher strategies--- deterministic search sequences that produce the same optimal expected total search time regardless of where the Hider hides---has been largely overlooked in the literature.
In the special case with $q_i=1/2$ and $t_i=1$ for $i \in [n]$, \cite{Ruckle} shows that the search sequence $n,n-1,\ldots,1,(1,2,\ldots,n)$ is optimal.
\cite{Ruckle} also shows that if $n=2$ and $q_1=q_2 \ge 0.8$, then there does not exist an optimal search sequence.
In this section, we prove much more general conditions for the existence of optimal pure strategy in Subsection~\ref{sec:existence}, and also conditions for the non-existence of optimal pure strategy in Subsection~\ref{sec:non-existence}.
These new findings allow us to extend the results in \cite{Ruckle} substantially.



To begin, we say that two search strategies $\theta_1$ and $\theta_2$ are {\em equivalent} if their expected search time against each box $i \in [n]$ is the same. That is,
\[
u(i,\theta_1)=u(i,\theta_2)
\]
for each $i \in [n]$.

\subsection{Existence of Optimal Pure Strategies}
\label{sec:existence}
For a  search strategy $\xi$, let $T_i^\ell(\xi)$ denote the sum of the search times of the boxes opened under~$\xi$, up to and including the $\ell^{\text{th}}$ search of box $i$, for $\ell=1,2,\ldots$. For example, if $\xi=1,2,2,1,\ldots$, then $T_1^1(\xi) = t_1$ and $T_{1}^2(\xi)=2t_1+2t_2$. 
Recalling $r_i = 1 - q_i$ is the overlook probability,
the payoff function $u(i,\xi)$ can be written as
\begin{align}
u(i,\xi)= \sum_{\ell=1}^\infty r_i^{\ell-1}q_i T_i^\ell(\xi). \label{eq:payoff}
\end{align}

To establish the existence of optimal pure strategies, we will focus our attention on the cases in which the Hider's equalizing strategy $\mathbf{p}^*$ is optimal.
Because we assume there are some positive coprime integers $k_1,\ldots,k_n$ and some $r >0$ such that $r_i^{k_i}=r$ for all $i \in [n]$, a Gittins search sequence against $\mathbf{p}^*$ can be broken up into cycles, where each cycle consists of $k_i$ searches in box $i$, $i \in [n]$, and has length $T = \sum_{i=1}^n k_i t_i$.
Taking advantage of this observation, the next lemma develops the right-hand side of~(\ref{eq:payoff}) further if $\xi$ is a Gittins search sequence against the Hider's equalizing strategy $\mathbf{p}^*$.



\begin{lemma} \label{lem:payoff}
Let $\xi_1=(s_1),\ldots,\xi_M=(s_M)$ denote $M$ distinct consistent Gittins search sequences against $\mathbf{p}^*$, where each subsequence $s_j$, $j=1,\ldots,M$, consists of $k_i$ searches in box $i$, for $i \in [n]$. 
Let $x_1,x_2,\ldots$ be a sequence taking values in $[M]$ and let $\xi$ be the Gittins search sequence $s_{x_1},s_{x_2},\ldots$.
The following results hold, where $I(\cdot)$ is the indicator function.
\begin{enumerate}[(i)]
    \item The payoff function $u(i,\xi)$ can be written as
\[
 u(i,\xi) = \frac{Tr}{1-r} + \sum_{j=1}^M w_i(s_j)\sum_{k=1}^\infty I(x_k=j)r^{k-1},
\]
where $w_i(s_j) = \sum_{\ell=1}^{k_i}  r_i^{\ell-1} q_i T_i^\ell(\xi_j)$ and $T=\sum_{i=1}^n k_it_i$.
\item If $\xi=(s)$ for some sequence $s$ then
\[
u(i,\xi) = \frac{Tr}{1-r} + \frac{w_i(s) }{1-r}.
\]
\end{enumerate}
\end{lemma}
\begin{proof}
For part (i), we observe that if $\ell=(k-1)k_i + \ell'$ for some $k=1,2,\ldots$ and some $1 \le \ell' \le k_i$, then $T_i^\ell(\xi_j) = (k-1)T + T_i^{\ell'}(\xi_j)$. Therefore,~(\ref{eq:payoff}) reduces to
\begin{align*}
    u(i,\xi) &= \sum_{k=1}^\infty \sum_{\ell=1}^{k_i}\left((k-1)T + \sum_{j=1}^M I(x_k=j)T_i^{\ell}(\xi_j) \right)r_i^{(k-1)k_i+(\ell-1)}q_i \\
     &= Tq_i\sum_{\ell=1}^{k_i}r_i^{\ell-1}\sum_{k=1}^\infty(k-1)r^{k-1} + \sum_{k=1}^\infty \sum_{\ell=1}^{k_i} \sum_{j=1}^M I(x_k=j)  r_i^{\ell-1}q_i T_i^{\ell}(\xi_j) r^{k-1} \text{ \ \ (since $r = r_i^{k_i}$)} \\
     &=Tq_i \frac{1-r_i^{k_i}}{1-r_i}\frac{r}{(1-r)^2}+ \sum_{k=1}^{\infty} \sum_{j=1}^M I(x_k=j) w_i(s_j) r^{k-1} \\ 
    &=\frac{Tr}{1-r}+ \sum_{j=1}^M w_i(s_j) \sum_{k=1}^\infty  I(x_k=j)  r^{k-1}
\end{align*}
For part (ii), we simply set $\xi_1=(s)$ and $x_k=1$ for all $k$. Applying part (i), we then have that
\[
u(i,\xi) = \frac{Tr}{1-r} + w_i(s)\sum_{k=1}^\infty r^{k-1} = \frac{Tr}{1-r} + \frac{w_i(s) }{1-r},
\]
which completes the proof.
\end{proof}

\bigskip

A key step to establish the optimality of a pure search strategy---or a search sequence---is to establish the equivalence between a pure search strategy and a mixed search strategy. 
Our next lemma shows that, for any Gittins search sequence against $\mathbf{p}^*$, it is always possible to construct an equivalent mixed strategy that consists of only consistent search sequences.


\begin{lemma} \label{lem:equiv}
    Let $\xi_1=(s_1),\ldots,\xi_{M} =(s_{M})$ denote $M \le n!$ distinct consistent Gittins search sequences against $\mathbf{p}^*$ and let $x_1,x_2,\ldots$ denote a sequence taking values in $[M]$. The search sequence $\xi^* \equiv s_{x_1},s_{x_2},\ldots$ is equivalent to the mixed Searcher strategy $\theta$ that chooses $\xi_j$ with probability
    \[
    \theta(\xi_j) = \sum_{k=1}^\infty I(x_k=j) (1-r)r^{k-1},
    \]
    for $j \in [M]$.
\end{lemma}
\begin{proof}
    First note that $\theta$ is a well-defined Searcher strategy, since
    \[
    \sum_{j=1}^{M} \theta(\xi_j) = \sum_{k=1}^\infty (1-r)r^{k-1} = 1.
    \]
    By Lemma~\ref{lem:payoff}, part (ii), the expected payoff of box $i$ against $\xi_j$ is
\begin{align*}
     u(i,\xi_j)  = \frac{Tr}{1-r}+ \frac{ w_i(s_j)}{1-r}.
\end{align*} 
Hence, the expected search time of box $i$ under $\theta$ is  
\begin{align*} 
    u(i,\theta) = \sum_{j=1}^{M} \theta(\xi_j) u(i,\xi_j)  = \frac{Tr}{1-r}+ \sum_{j=1}^M \frac{\theta(\xi_j) }{1-r} w_i(s_j) 
\end{align*}
Applying Lemma~\ref{lem:payoff}, part (i), the expected search time of box $i$ against $\xi^*$ is
\begin{align*}
    u(i,\xi^*)
    &= \frac{Tr}{1-r} + \sum_{j=1}^{M} w_i(s_j)\sum_{k=1}^\infty I(x_k=j)r^{k-1}\\ 
    &= \frac{Tr}{1-r}+  \sum_{j=1}^M w_i(s_j) \frac{\theta(\xi_j)}{1-r},
\end{align*}
by definition of $\theta$. Hence, $u(i,\theta) =   u(i,\xi^*)$.
\end{proof}

\bigskip

Before presenting our main result, we present another lemma that will help us to show it is possible---under certain conditions---to construct a search sequence that is equivalent to a mixed search strategy that mixes between consistent Gittins search sequences against $\mathbf{p}^*$. Though this lemma is crucial to our main result, it stands alone as a lemma in probability theory. We were not able to find it in the literature, and it may be of independent interest.


\begin{lemma}\label{lem:seq}
Let $M$ be a positive integer. If $1-1/M\le r<1$, then for any $\lambda_1,\ldots,\lambda_M \ge 0$ with $\lambda_1+\cdots+\lambda_M=1$, there exists a sequence $x_1,x_2,\ldots$ taking values in $[M]$ such that
\begin{align}
\sum_{k=1}^\infty I(\{x_k=j\}) (1-r)r^{k-1} = \lambda_j, \label{eq:partition}
\end{align}
for each $j\in[M]$, where $I$ is the indicator function. 
Moreover, if $r^d \ge 1-\min_j \lambda_j$, then there are at least $M^d$ sequences that satisfy condition~(\ref{eq:partition}). 
\end{lemma}
\begin{proof}
We define the sequence $x_1,x_2,\ldots$ recursively and will prove by induction on $K$ that for each $K=0,1,\ldots$, we have
\begin{align}
\sum_{k = 1}^K I(\{x_k=j\}) (1-r)r^{k-1} \le \lambda_j, \label{eq:hyp}
\end{align}
for each $i\in[M]$. This is trivially true for $K=0$. Suppose~(\ref{eq:hyp}) is true for some $K \ge 0$, and note that
\begin{align}
\sum_{j =1}^M \left( \lambda_j - \sum_{k =1}^K I(\{x_k=j\}) (1-r)r^{k-1} \right)  = 1 - \sum_{k=1}^K (1-r)r^{k-1} = r^K. \label{eq:sum}
\end{align}
By the induction hypothesis, every term in the outer sum on the left-hand side of~(\ref{eq:sum}) is non-negative, and it follows that for some $j\in[M]$, we have 
\[
\lambda_{j} - \sum_{k =1}^K I(\{x_k=j\})(1-r)r^{k-1} \ge r^K/M \ge (1-r)r^K,
\]
since $r>1-1/M$. Then we can take $x_{K+1}=j$, and~(\ref{eq:hyp}) is clearly satisfied for $K+1$. Equation~(\ref{eq:partition}) follows from taking limits in~(\ref{eq:sum}) as $K \rightarrow \infty$.  

If $1-r \leq \min_j \lambda_j$, there are $M$ choices for $x_1$, and consequently there are at least $M$ distinct sequences. In general, if $1- r^d = \sum_{k=1}^d (1-r) r^{k-1} \leq \min_j\lambda_j$ for some $d \geq 1$, there are $M$ choices for each $x_i$, $i= 1, \ldots, d$, so at least $M^d$ distinct sequences meet condition (\ref{eq:partition}).
\end{proof}

\bigskip

It is interesting to point out that Lemma~\ref{lem:seq} relies on the axiom of choice in the induction step.

Also, note that we cannot strengthen the bound $r \ge 1-1/M$ in Lemma~\ref{lem:seq}, since for any $r<1-1/M$, if we take $\lambda_1=\cdots=\lambda_M=1/M$, then for any choice $x_1=i$, we have
\[
\sum_{k=1}^\infty I(\{x_k=j\}) (1-r)r^{k-1} \ge (1-r) > 1/M=  \lambda_{j}. 
\]

We now use Lemmas~\ref{lem:equiv} and~\ref{lem:seq} to prove our main results in the next theorem, which gives sufficient conditions under which an optimal pure strategy for $G$ can be found.
\begin{theorem} \label{lem:m_seq_diff_prob}
    Let $\xi_1, \ldots, \xi_M$ be $M$ consistent Gittins search sequences against $\mathbf p^*$ and suppose $\theta$ is some Searcher mixed strategy  with support $\{\xi_1, \ldots, \xi_M\}$. 
    If $r= r_i^{k_i} \geq 1- 1/M$ for all $i\in [n]$, then
    \begin{enumerate}[(i)]
        \item  there exists some search sequence $\xi^*$ such that
    \begin{align}
    u(i, \xi^*)= u(i, \theta) \label{eq:pure}
    \end{align}
    for all $i \in [n]$, \item moreover if $r^d \geq 1-\min_j\theta(\xi_j)$ for some $d= 1, \ldots$ then there are at least $M^d$ distinct search sequences satisfying \eqref{eq:pure}, and 
   
    \item in particular, if $\mathbf{p}^*$ is optimal for the Hider and $r \ge 1-1/n$, then the Searcher has an optimal pure strategy.
    \end{enumerate}
\end{theorem}
\begin{proof}
To prove part (i) of the theorem, let $\lambda_j = \theta(\xi_j)$ for $j \in [M]$, so that $\sum_{j=1}^M \lambda_j =1$.
By definition, we have that
\[
u(i, \theta) = \sum_{j=1}^M \lambda_j u(i, \xi_j).
\]
Let $x_1,x_2,\ldots$ be a sequence taking values in $[M]$ such that
\begin{align}\label{eq:qk}
    \sum_{k=1}^\infty I(x_k=j) (1-r)r^{k-1} = \lambda_j
\end{align} 
for all $j \in [M]$. By Lemma~\ref{lem:seq}, such a sequence exists because $r\geq 1-1/M$. 

Since each $\xi_j$ is a consistent Gittins search sequence against $\mathbf{p}^*$, it can be written as $\xi_j=(s_j)$, for some finite subsequence $s_j$. Let $\xi^*$ be the pure Searcher strategy $s_{x_1}, s_{x_2}, \ldots$ 
Then part (i) of the theorem follows from Lemma~\ref{lem:equiv}. 

For part (ii), if $1- r^d \leq \min_j \theta(\xi_j)$ for some $d\geq 1$, we know from the second statement of Lemma~\ref{lem:seq} that at least $M^d$ distinct sequences $x_1, x_2, \ldots$ exist. So, by the construction of $\xi^*$, it is straightforward that there are at least $M^d$ distinct pure strategies $\xi^*$ which satisfy (\ref{eq:pure}).

Part (iii) of the theorem follows from the fact that if $\mathbf{p}^*$ is optimal, then by Theorem~\ref{thm:CLG}, there is an optimal Searcher strategy $\theta$ that mixes between $n$ consistent Gittins search sequences against $\mathbf{p}^*$. So by part (i) of this theorem, there must be a Searcher pure strategy whose expected search time against each $i\in [n]$ is the same as that of $\theta$, and which is therefore optimal.
\end{proof}

\bigskip

We can use Theorem~\ref{lem:m_seq_diff_prob} to show that the Searcher has optimal pure strategies in cases of the game $G$ for which we know $\mathbf{p}^*$ is optimal.
As indicated in Theorem~\ref{lem:m_seq_diff_prob}, for there to exist an optimal pure search strategy, $r$---the probability of not finding the Hider in the first cycle of $k_i$ searches in box $i$, $i=1,\ldots,n$---needs to be sufficiently large.
Intuitively, with a large $r$, the subsequence chosen in the first cycle plays a less important role, which makes it possible to achieve \eqref{eq:pure} by prioritizing the other boxes in the following cycles.

\begin{corollary} \label{cor:equal-det}
If the detection probabilities $q_1,\ldots,q_n$ are all equal to some $q \le 1/n$, then the Searcher has an optimal pure strategy. In addition, if $1- (1-q)^d \leq \min_j t_j/T$ for some $d \geq 1$, then there are at least $n^d$ distinct optimal pure strategies. 
\end{corollary}
\begin{proof}
By Theorem~\ref{thm:q-equal}, the Hider's equalizing strategy $\mathbf{p}^*$ is optimal and the Searcher strategy $\theta$ which chooses each sequence $\xi_i= (i+1, i+2, \ldots, n, 1, 2, \ldots, i-1)$, $i = 1, \ldots, n$, with probability $t_i/T$ is optimal. Since $r= r_1=\cdots=r_n =1-q \ge 1-1/n$, the existence of an optimal pure search strategy follows Theorem~\ref{lem:m_seq_diff_prob}, part (iii). Moreover, by Theorem~\ref{lem:m_seq_diff_prob} part (ii), if $1-(1-q)^d= 1-r^d \leq \min_i \theta(\xi_i) = \min_i t_i/T$, there exist at least $n^d$ optimal pure strategies.
\end{proof}

\begin{corollary} \label{cor:equal-unit}
    If the detection probabilities are all equal to some $q \le 1/2$ and the search times are all equal, then the Searcher has an optimal pure strategy. Furthermore, if $1-(1-q)^d \le 1/2$ for some $d \geq 1$, then there are at least $2^d$ distinct optimal pure strategies.
\end{corollary}
\begin{proof}
Again, by Theorem~\ref{thm:q-equal}, the Hider's equalizing strategy $\mathbf{p}^*$ is optimal. By Proposition~\ref{prop:equal}, the Searcher has an optimal strategy that chooses two pure strategies each with probability $1/2$. Since $r= r_1=\cdots=r_n =1-q \ge 1/2$, the existence of an optimal pure search strategy follows Theorem~\ref{lem:m_seq_diff_prob}, part (iii). Moreover, by Theorem~\ref{lem:m_seq_diff_prob} part (ii), if $1-(1-q)^d= 1-r^d \leq \min_i \theta(\xi_i) = 1/2$, there exist at least $2^d$ optimal pure strategies.
\end{proof}

\begin{corollary}
\label{cor:n=2}
If there are $n=2$ boxes with $t_1=t_2$ and $r_1^k=r_2^{k+1}\ge 1/2$ for some $k=1,\ldots,12$, then the Searcher has an optimal pure strategy.
\end{corollary}
\begin{proof}
    By Theorem~\ref{thm:RG78}, the Hider's equalizing strategy $\mathbf{p}^*$ is optimal, so the corollary follows immediately from Theorem~\ref{lem:m_seq_diff_prob}, part (iii).
\end{proof}

We illustrate Corollary~\ref{cor:equal-det} with the following example.

\begin{example}
Consider $n=3$ with $q_1= q_2= q_3 = q= 0.2$ and $t_1= 2.00$, $t_2 = 2.88$, $t_3 = 5.12$. Let $s_1 = 1, 2, 3$, $s_2 = 2, 3, 1$ and $s_3 = 3, 1, 2$. We have $r= 1-q = 0.8$ and $T= t_1+t_2+t_3=10$. By Corollary \ref{cor:equal-det}, the Searcher has an optimal strategy. We can construct the optimal strategy as follows.

First, following Lemma \ref{lem:seq}, we construct a sequence $\sigma = x_1, x_2, \ldots$ taking values in $\{1, 2,3\}$ such that $\sum_{k=1}^\infty I(x_k= i) (1-r)r^{k-1}= t_i/T$ for $i= 1, 2, 3$. Let $\delta_i(
K
) = \sum_{k=1}^K I(x_k= i) (1-r)r^{k-1}$ for $K= 0, 1, \ldots$ Note that $\delta_i(0) = 0$ for all $i$. For any $K= 1, 2, \ldots$, we define $x_K$ recursively such that if $\delta_i(K-1)+ (1-r)r^{K-1} \leq t_i/10$, then $x_K = i$.
\begin{itemize}
\item For $K=1$, we have $\delta_i(0)+ (1-r)r^{K-1}= 1-r= 0.2 \leq t_i/T$ for all $i$, so $x_1$ can be 1, 2 or 3.
\item For $K=2$, if we let $x_1 = 1$, it is easy to check that $x_2$ just can be $2$ or $3$.
\item Continuing this way, we obtain the sequence $\sigma = 1,2,2,3,(3)$.
\end{itemize}

Second, let $\xi^*= s_1, s_2, s_2, s_3, (s_3) = 1,2,3, 2,3,1,2,3,1, 3,1, 2, (3,1,2)$. Then, $\xi^*$ is optimal by Theorem \ref{lem:m_seq_diff_prob}. 
Note that if we change the choice of $x_1$ or $x_2$, we will get a different sequence $\sigma$ which leads to another different optimal pure strategy.
\end{example}

\subsection{Non-existence of Optimal Pure Strategies} \label{sec:non-existence}

We now consider what conditions ensure that there will be no optimal pure strategy for the Searcher. As in the previous section we focus on the case that the Hider's equalizing strategy $\mathbf{p}^*$ is optimal.

For a particular instance of the game $G$, suppose $\theta$ is an optimal Searcher strategy that is a mixture of distinct consistent Gittins search sequences $\xi,\ldots,\xi_M$, for some $M \le n!$.
Write $\lambda_{\max}(\theta) = \max_{1 \le j \le M} \theta(\xi_j)$, and let $\lambda^*$ denote the supremum of $\lambda_{\max}(\theta)$ over all such optimal Searcher strategies.


\begin{lemma} \label{lem:r_theta}
    If $\mathbf{p}^*$ is optimal and $r_1^{k_1}=\cdots=r_n^{k_n}=r < 1- \lambda^*$, then there does not exist an optimal pure strategy for the Searcher.
\end{lemma}
\begin{proof}
Suppose there is an optimal pure Searcher strategy $\xi^*$. Let $\xi_1=(s_1),\ldots,\xi_{n!}=(s_{n!})$ be the consistent Gittins search sequences against $\mathbf{p}^*$. Since $\mathbf{p}^*$ is optimal, $\xi^*$ must be a Gittins search sequence against $\mathbf{p}^*$. Hence, $\xi$ can be written as $s_{x_1}, s_{x_2}, \ldots$, where $x_k \in [n!]$, for $k=1,2,\ldots$. Let $\theta$ be the search strategy that chooses $\xi_j$ with probability 
\[
\theta(\xi_j) = \lambda_j \equiv \sum_{k=1}^\infty I(x_k = j) (1-r) r^{k-1},
\]
for $j\in[n!]$. 

By Lemma~\ref{lem:equiv}, strategies $\theta$ and $\xi^*$ are equivalent, so $\theta$ is optimal. But 
\[
\lambda^* \ge \lambda_{\max}(\theta) = \max_j \lambda_j \ge \lambda_{x_1} = 1-r + \sum_{k=2}^\infty I(x_k = x_1) (1-r) r^{k-1} \ge 1-r > \lambda^*,
\]
a contradiction.
\end{proof}

\bigskip

Lemma~\ref{lem:r_theta} does not give much clue as to how one might calculate $\lambda^*$, and it is not even clear that $\lambda^*$ is strictly less than 1. We now consider the two classes of solutions for which we know that $\mathbf{p}^*$ is optimal and show that there is a non-empty interval of values of $r$ for which there is no optimal pure strategy.

\begin{theorem} \label{thm:nopure}
For equal detection probabilities $q_1=\cdots=q_n = q$, there is no optimal pure strategy for $q>q^*$, where
\begin{align}
q^* \equiv 1- \frac{1}{T(T- t_{\min})}\sum_{\{a, b\} \subset [n]} t_at_b, \label{eq:q-ineq}
\end{align}
and $t_{\min} = \min_i t_i$.
\end{theorem}
\begin{proof}
Since the detection probabilities are equal, according to Theorem~\ref{thm:q-equal} the equalizing strategy~$\mathbf{p}^*$ is optimal. The consistent Gittins search sequences against $\mathbf{p}^*$ are $\xi_1=(\sigma_1),\ldots,\xi_{n!}=(\sigma_{n!})$, where $\sigma_1,\ldots,\sigma_{n!}$ are the permutations of $[n]$. Let $\theta$ be an optimal Searcher strategy that mixes between $\xi_1,\ldots,\xi_{n!}$ and let $\lambda_j=\theta(\xi_j),j \in [n!]$. 

Without loss of generality suppose the last term of $\sigma_1$ is $i$.
Using Lemma~\ref{lem:payoff}, part (ii) and the fact that $1-r=q$, the expected search time of box $i$ under $\theta$ can be written as
\[
u(i,\theta) =  \frac{T(1-q)}{q} + \sum_{j=1}^{n!} \frac{\lambda_j w_i(\sigma_j) }{q}.
\]
Since $\theta$ is optimal, $u(i,\theta)$ is equal to the value $V$ of the game, which is given in Theorem~\ref{thm:q-equal}. Hence,
\begin{align}\label{eq:V_u}
    0=V- u(i,\theta)  &= \left(\frac{T}{q}- \frac{1}{T} \sum_{\{a, b\} \subset [n]} t_at_b\right) - \left(\frac{T(1-q)}{q} + \sum_{j=1}^{n!} \frac{\lambda_j w_i(\sigma_j)}{q}\right) \nonumber \\
    &= T - \frac{1}{T} \sum_{\{a, b\} \subset [n]} t_at_b - \sum_{j=1}^{n!} \frac{\lambda_j w_i(\sigma_j)  }{q}
\end{align} 
Since the last term of $\sigma_1$ is $i$, we must have $w_i(\sigma_1) = qT$. Also, $w_i(\sigma_j) \geq qt_i$ for all $j \neq 1$, so
\[
    \sum_{j=1}^{n!}\frac{\lambda_j w_i(\sigma_j)}{q} \geq  \lambda_1  T +  \sum_{j\neq 1} \lambda_j t_i =
     (T-t_i)\lambda_1+t_i.
\]
Combining the preceding with~(\ref{eq:V_u}), we obtain 
\[
\lambda_1 \le 1- \frac{1}{T(T-t_i)}\sum_{\{a, b\} \subset [n]} t_at_b \le     q^*,
\]
by definition of $q^*$. By a similar analysis, the inequality above holds when $\lambda_1$ is replaced by any $\lambda_j, 2\le j \le n!$, so $\lambda_{\max}(\theta) \le q^*$ and $\lambda^* \le q^*$, by definition of $\lambda^*$.

Hence, if $q > q^*$, then $q > \lambda^*$, so $r=1-q <1-\lambda^*$, and by Lemma~\ref{lem:r_theta}, there is no optimal pure strategy.
\end{proof}

\bigskip

Note that $q^*$ is always at least $1/2$, because
\[
\sum_{\{a, b\} \subset [n]} t_at_b  = \frac{1}{2} \left( T^2 - \sum_{i\in [n]} t_i^2 \right) 
\le  \frac{1}{2} \left( T^2 - \sum_{i\in [n]} t_i t_{\min} \right)  = \frac{1}{2} T(T-t_{\min}).
\]
It is straightforward to verify that in the case of equal search times as well as equal detection probabilities, $q^*$ is equal to $1/2$, so we obtain the following corollary of Theorem~\ref{thm:nopure}.
\begin{corollary} \label{cor:equal}
For equal detection probabilities $q$ and equal search times, there is no optimal pure strategy for $q>1/2$.
\end{corollary}

The bound $q^*$ is tight in this case, since there is an optimal pure strategy for all $q \le 1/2$, by Corollary~\ref{cor:equal-unit}.
Intuitively, if $q > 1/2$ then the first few searches affect the expected total search time so significantly that it is necessary for the Searcher to use a mixed strategy to achieve optimality.
If $q \leq 1/2$, then it is possible to achieve optimality with a single search sequence---a pure strategy---by prioritizing the other boxes later on to neutralize the advantage of the first few boxes in the search sequence.

For general search times, the bound of $q^*$ may still be tight.  For example, for the three-box game, $q_1 = q_2= q_3=q$, $t_{\min}= t_1$, and $t_2^2= t_1t_3$, we have $q^* = \frac{t_2^2+t_2t_3+t_3^2}{(t_1+t_2+t_3)(t_2+t_3)}$. 
It is straightforward to show that the search sequence $3,2,1, (1,2,3)$ is optimal when $q= q^*$.

We now turn to the case that $t_1=t_2=1$ and $r=r_1^k=r_2^{k+1}$ for $k= 1, \ldots, 12$. The Hider's equalizing strategy $\mathbf{p}^*$ is optimal, by Theorem~\ref{thm:RG78}. There are two consistent Gittins search sequences $\xi_1= (s_1)$, $\xi_2= (s_2)$ against $\mathbf{p^*}$ such that $s_1= 1,2,\ldots$, $s_2= 2,1,\ldots$ and the last $2k-1$ terms of $s_1$ and $s_2$ are the same. For example, $k=3$, then $s_1= 1,2,2,1,2,1,2$ and $s_2= 2,1,2,1,2,1,2$. It is straightforward to verify that $w_1(s_1)= w_1(s_2)-q_1$ and $w_2(s_1)= w_2(s_2)+q_2$

Let $\theta$ be an optimal Searcher strategy that chooses $\xi_j$ with probability $\lambda_j$ for $j=1, 2$. By Lemma~\ref{lem:payoff}, the expected search time of box $1$ under $\theta$ can be written as
\begin{align*}
    u(1, \theta)&= \frac{Tr}{1-r}+\frac{1}{1- r}\big(\lambda_1w_1(s_1)+\lambda_2w_1(s_2)\big) \\
    &=\frac{Tr}{1-r}+\frac{1}{1- r}\big(\lambda_1(w_1(s_2)-q_1)+\lambda_2w_1(s_2)\big)
    \\&=\frac{Tr}{1-r}+\frac{1}{1- r}\big(w_1(s_2)-q_1\lambda_1\big)
\end{align*} 
Similarly, the expected search time of box 2 under $\theta$ can be written as
\begin{align*}
    u(2, \theta)&= \frac{Tr}{1-r}+\frac{1}{1- r}\big(\lambda_1w_2(s_1)+\lambda_2w_2(s_2)\big)\\&=\frac{Tr}{1-r}+\frac{1}{1- r}\big(\lambda_1(w_2(s_2)+q_2)+\lambda_2w_2(s_2)\big)\\&=\frac{Tr}{1-r}+\frac{1}{1- r}\big(w_2(s_2)+q_2\lambda_1\big)
\end{align*}
Since $\theta$ is optimal, we have $u(1,\theta)= u(2,\theta)$ or $w_1(s_2)-q_1\lambda_1= w_2(s_2)+q_2\lambda_1$.
Therefore,
\begin{align}
\label{eq:lambda1}
    \lambda_1 = \frac{w_1(s_2)-w_2(s_2)}{q_1+q_2}= \frac{w_1(s_2)-w_2(s_2)}{2-r_1-r_2}.
\end{align} 

By Lemma~\ref{lem:r_theta}, there is no optimal pure strategy if $r<1-\max(\lambda_1, \lambda_2)$ or equivalently,
\begin{align}
\label{eq:lambda1_r}
r < 1 - \lambda_1 \qquad \text{and} \qquad r < 1 - \lambda_2 = \lambda_1
\end{align} 
Solving (\ref{eq:lambda1})-(\ref{eq:lambda1_r}) computationally for $k=1,\ldots,12$, $r= r_1^k= r_2^{k+1}$, we get $r^*$ such that there is no optimal pure strategy for $r<r^*$ as follows.
\begin{center}
\begin{tabular}{||c c|| c c||} 
 \hline
 $k$ & $r^*$ & $k$ & $r^*$ \\ [1ex] 
 \hline\hline
 1 & 0.216757 & 7 & 0.399087 \\ 
 \hline
 2 & 0.233196 & 8 & 0.426468 \\
 \hline
 3 & 0.267022 & 9 & 0.451679 \\
 \hline
 4 & 0.302809 & 10 & 0.474914 \\
 \hline
 5 & 0.337219 & 11 & 0.496363 \\
 \hline
 6 & 0.369366 & 12 & 0.023657 \\ [1ex] 
 \hline
\end{tabular}
\end{center} Note that for any arbitrary $t_1=t_2=t$, the bounds $r^*$ do not change.

\section{The Case of Two Boxes} \label{sec:two}
In this section we restrict our attention to the case of $n=2$ boxes. In Subsection~\ref{sec:opt-2boxes}, we examine conditions under which there are (or are not) optimal pure strategies.  In Subsection~\ref{sec:best-2boxes}, we pose the question: if there is no optimal pure strategy, what is the best possible pure strategy we can find?

\subsection{Optimal Pure Strategies for Two Boxes} \label{sec:opt-2boxes}

We first investigate the existence of optimal pure Searcher strategies for the two-box game, $n=2$, making the standing assumption that $t_1 \leq t_2$. We start with a general result that follows from Theorem~\ref{lem:m_seq_diff_prob}.

\begin{theorem} \label{thm:n=2}
If $n=2$ and $r_1^{k_1}=r_2^{k_2} \ge 1/2$ then the Searcher has an optimal pure strategy.
\end{theorem}
\begin{proof}
Let $\mathbf{p}$ denote the optimal Hider strategy and let $\mathbf{p}^*$ denote the Hider's equalizing strategy. If $\mathbf{p}=\mathbf{p}^*$, then the theorem follows from Theorem~\ref{lem:m_seq_diff_prob}, so assume that $\mathbf p\neq \mathbf p^*$.

By Theorem~\ref{thm:CLG}, there is an optimal Searcher strategy $\theta$ whose support is two consistent Gittins search sequences $\xi_1$ and $\xi_2$ against $\mathbf{p}$. If $\xi_1=\xi_2$, then the there is nothing to prove, so assume that $\xi_1 \neq \xi_2$. Then $\xi_1$ and $\xi_2$ must be identical up until the $j$th term, for some $j$, at which point, the Searcher is indifferent between searching the two boxes: that is, the new conditional hiding distribution is $\mathbf{p}^*$. So each $\xi_i, i=1,2$ can be written $\xi_i=s,\xi'_i$ for some finite sequence $s$ and some consistent Gittins search sequences $\xi'_1$ and $\xi'_2$ against $\mathbf{p}^*$. Let $\theta'$ be the search strategy that chooses $\xi'_i$ with probability $\theta(\xi_i)$ for $i=1,2$. By Theorem~\ref{lem:m_seq_diff_prob}, there is a pure strategy $\xi^*$ such that $u(i,\xi^*)=u(i,\theta')$ for $i=1,2$. Then, if we set $\xi^{**}=s,\xi^*$, it is easy to see that $u(i,\xi^{**})=u(i,\theta)$ for $i=1,2$, so $\xi^{**}$ is optimal.
\end{proof}

\bigskip

We now consider separately the special case when $q_1=q_2=q$. By Corollary~\ref{cor:equal-det} (or Theorem~\ref{thm:n=2}), the Searcher has an optimal pure strategy for $q \le 1/2$ and by Theorem~\ref{thm:nopure}, the Searcher has no optimal pure strategy for $q>q^*$, where $q^*$ reduces to $t_2/(t_1+t_2)$. This bound is tight, since a straightforward calculation shows that the strategy $2,1,(1,2)$ is optimal for the Searcher when $q=t_2/(t_1+t_2)$. Also, note that $2,1, (1,2)$ is a unique optimal pure strategy when $q= t_2/(t_1+t_2)$ and $t_2> t_1$. Indeed, any optimal pure strategy must be a Gittins sequence $\xi = s_{x_1}, s_{x_2} \ldots$ where each $x_k \in \{1,2\}$, $s_1 = 1, 2$, and $s_2 = 2,1$. Suppose $s_{x_1}= 1,2$, then it is obvious $u(2, \xi) \geq u(2,\xi')$ where $\xi' = 1, 2, (2,1)$. However, $$u(2, \xi) \geq u(2,\xi') = (t_1+t_2)/q- t_1^2/(t_1+t_2) > (t_1+t_2)/q- (t_1 t_2)/(t_1+t_2) =V.$$ So, $s_{x_1}$ must be $2,1$. Then, it is straightforward that $s_{x_k} = 1,2$ for all $k=2, \ldots$, otherwise $u(1, \xi) >V$. Hence, the strategy $2,1,(1,2)$ is unique.

One might conjecture that there is an optimal pure strategy for all the values of $q$ between $1/2$ and $t_2/(t_1+t_2)$.
However, it turns out that this conjecture is not true in general.
We show that, at least for some choices of $t_1$ and $t_2$, there is a non-empty interval of values of $q$ between $1/2$ and $t_2/(t_1+t_2)$ for which there is no optimal pure strategy. For example, if $t_2/(t_1+t_2) = 0.8$, then there is no optimal pure strategy for $0.553<q<0.723$. In general, we have the following result.
\begin{lemma}
Consider the case with $n=2$ and equal detection probabilities $q_1=q_2=q$ and search times $t_1$ and $t_2$ with $t_1 \le t_2$.
If 
\[
\frac{t_1}{t_1+t_2} < \frac{1}{4} \qquad \text{and} \qquad 1- \sqrt{\frac{t_1}{t_1+t_2}}< q <  \frac{1}{2} \left(\sqrt{1-\frac{4t_1}{t_1+t_2}}+1 \right),
\]
then there is no optimal pure strategy.
\end{lemma}
\begin{proof}
We first derive necessary conditions for a pure strategy to be optimal. By Theorem~\ref{thm:q-equal}, the Hider's equalizing strategy $\mathbf{p}^*$ is optimal and the value of the game is
\begin{align}
V=\frac{t_1+t_2}{q}-\frac{t_1 t_2}{t_1+t_2}. \label{eq:V}
\end{align}
Let $\xi$ be an optimal pure strategy. It must be a Gittins search sequence against $\mathbf{p}^*$, so $\xi$ must be of the form $\xi=s_{x_1},s_{x_2},\ldots$, where $s_1=(1,2)$ and $s_2=(2,1)$ and $x_k \in \{1,2\}$. Let $\gamma_k=I(x_k=1)$ (in other words, $\gamma_k$ is equal to 1 if $s_{x_k}=(1,2)$, otherwise it is equal to 0). Note that $w_1(s_1)=qt_1$ and $w_1(s_2)=q(t_1+t_2)$. Similarly for $w_2(s_1)$ and $w_2(s_2)$.

By Lemma~\ref{lem:payoff}, part (i),
\begin{align}
    u(1,\xi) & = \frac{T(1-q)}{q} + \sum_{k=1}^\infty q t_1 \gamma_k r^{k-1} + \sum_{k=1}^\infty q(t_1+t_2)(1-\gamma_k)r^{k-1} \nonumber \\
    & = \frac{T(1-q)}{q} + q(t_1+t_2)\sum_{k=1}^\infty  r^{k-1} - qt_2\sum_{k=1}^\infty \gamma_k r^{k-1} \nonumber\\
    & = \frac{T}{q} - q t_2 \sum_{k=1}^\infty \gamma_k r^{k-1}. \label{eq:u1}
\end{align}
Similarly, 
\begin{align}
u(2,\xi) = \frac{T}{q} - q t_1 \sum_{k=1}^\infty (1- \gamma_k )r^{k-1} \label{eq:u2}
\end{align}
Since $\xi$ is optimal, we must have $V=u(1,\xi)=u(2,\xi)$, and setting~$(\ref{eq:V})$,~(\ref{eq:u1}) and~(\ref{eq:u2}) to be equal gives
\begin{align}\label{eq4}
    \begin{cases}
        q\sum_{k=1}^\infty \gamma_k r^{k-1}  =  t_1/(t_1+t_2),\\
        q\sum_{k=1}^\infty  (1-\gamma_k)  r^{k-1} = t_2/(t_1+t_2).
    \end{cases}
\end{align} 

Let $y_1= 1- \sqrt{\frac{t_1}{t_1+t_2}}$ and $y_2= \frac{1}{2}(\sqrt{1-\frac{4t_1}{t_1+t_2}}+1)$. Observe that $1/2<y_1 < y_2 < t_2/(t_1+t_2)$ when $t_1/(t_1+t_2)<1/4$. 
It is easy to see that $\gamma_1$ must be $0$ because if $\gamma_1 = 1$, then $q\sum_{k=1}^\infty \gamma_k  r^{k-1} \ge q >1/2 \ge t_1/(t_1+t_2)$, contradicting~(\ref{eq4}). 

We consider two cases for $\gamma_2$.  First, if $\gamma_2=0$ then $q\sum_{k=1}^\infty (1-\gamma_k) r^{k-1}  \ge q(1+r)=q(2-q)$. It is straightforward to check that $q(2-q)>t_2/(t_1+t_2)$ for $1- \sqrt{\frac{t_1}{t_1+t_2}}<q<  \frac{1}{2}(\sqrt{1-\frac{4t_1}{t_1+t_2}}+1)$.
Second, if $\gamma_2=1$, then $q\sum_{k=1}^\infty  \gamma_k r^{k-1} \ge qr=q(1-q)$, and it is straightforward to check that $q(1-q) > t_1/(t_1+t_2)$ for $1- \sqrt{\frac{t_1}{t_1+t_2}}<q<  \frac{1}{2}(\sqrt{1-\frac{4t_1}{t_1+t_2}}+1)$.
In either case, we have a contradiction, so we conclude that there does not exist an optimal pure strategy.
\end{proof}

\subsection{Best Pure Strategies} \label{sec:best-2boxes}
We have seen that for equal search times and equal detection probabilities, there is no optimal pure strategy solution for $q>1/2$ when $n=2$.
A natural question to ask is what the \textit{best} pure strategy is for this range of values of $q$. 
More precisely, we wish to find a search sequence $\xi$ that minimizes $\max \{u(1,\xi),u(2,\xi)\}$. \cite{Ruckle} shows that for equal search times, there always exists $\xi$ such that $u(1,\xi)=u(2,\xi) = \cdots = u(n, \xi)$.
However, a search sequence that results in the same expected total search time regardless of where the Hider hides is not necessarily the best search sequence.
Below is a counterexample.

\begin{example}
\label{ex:best}
Consider $n=2$ with $q_1=q_2=0.5$ and $t_1=t_2=1$.
According to Corollary~\ref{cor:n=2}, the Searcher has an optimal pure strategy.
It is straightforward to verify that the value of the game is $7/2$ and the search sequence 
\[
1,2,(2,1)
\]
is optimal.
However, one can also verify that a different search sequence
\[
1,2,1,2,2,2,2,1,1,2,1,2,1,1,1,2,(2,1)
\]
results in an expected search time $923/256$ whether the Hider hides in box 1 or box 2.
In other words, it is possible that there exists more than one search sequence that results in the same expected search time regardless of where the Hider hides, so finding one such search sequence does not mean it is the best search sequence that minimizes $\max \{u(1,\xi),u(2,\xi)\}$.
\hfill $\Box$
\end{example}

\bigskip

While $u(1, \xi) = \cdots = u(n, \xi)$ is not a sufficient condition for $\xi$ to be the best pure strategy, it is clearly a necessary condition, for otherwise it is possible to make small tweaks to $\xi$ to reduce $\max_{i=1,\dots,n} u(i, \xi)$.
For the case with $n=2$, $q_1=q_2=q=2/3$ so that $r=1-q=1/3$, a simple calculation shows that the search sequence $\xi=1,2,(2,2,1,1)$ guarantees an expected search time of 
\[
u(1,\xi)=u(2,\xi)= 2 + r + \frac{r^2 + 3 r^3}{1-r^2} = \frac{31}{12}.
\]
We conjecture that this is the best pure strategy for $q=2/3$. Computational results indicate that any pure strategy with expected search time at most $31/12$ must begin with the sequence $1,2$ followed by 15 cycles of $2,2,1,1$ (or the same sequence with 1 and 2 swapped).
Below we generalize the conjecture to the case with $n=2$ and $q_1 = q_2 = (m-1)/m$, for $m=3, 4, \ldots$.

\begin{conjecture}
\label{co:best}
Suppose $n=2$ and $q_1 = q_2 =(m-1)/m$, for some integer $m \geq 3$. Let $r=1-q_1 = 1/m$.
The best pure strategy is $\xi=1,2,(2,2,\ldots,2,1,1,\ldots,1),$ in which the Searcher begins with $1, 2$, and then indefinitely repeats the same pattern that consists of $(m-1)$ searches in box 2 followed by $(m-1)$ searches in box 1.
The resulting expected search time is
\[
u(1,\xi)=u(2,\xi)=2 + r + \frac{\sum_{j=2}^{m-1} r^j + m r^m}{1-r^{m-1}} = 2 + \frac{1}{m} + \frac{m^{m-2} + m - 2}{(m-1)(m^{m-1}-1)}.
\]
\end{conjecture}

\bigskip

We note that for $m=2$, Conjecture~\ref{co:best} is true because in that case the search sequence $1,2, (2,1)$ is optimal as seen in Example~\ref{ex:best}.

\section{Conclusion} \label{sec:conclusion}
This paper studies a discrete search game  and investigates the existence of optimal pure strategies for the Searcher---a single deterministic search sequence that achieves the optimal expected total search time regardless of where the Hider hides.
An optimal pure search strategy has significant practical value because it is straightforward to execute without the need of randomization.
It would also avoid the potential criticism from lay persons that the chosen course of action does not turn out well even though the chosen action is properly selected from a set of pure strategies that compose the optimal mixed strategy.

For a two-person zero-sum game, typically if one player's optimal mixed strategy includes several pure strategies, the other player's optimal mixed strategy also includes the same number of pure strategies.
In our search game, the Hider's optimal mixed strategy always instructs the Hider to hide in each of the $n$ boxes with a strictly positive probability, so it may come as a surprise that the Searcher has an optimal pure strategy in several nontrivial cases.
In particular, because the first few boxes in a search sequence have a profound effect on the conditional expected total search time for each box, intuitively it is necessary for the Searcher to randomize the first few boxes in order to achieve optimality.
However, if the detection probabilities $q_i$, $i=1,\ldots,n$, are sufficiently small, then the search tends to take a long time with high probability, so the effect of the first few searches becomes less significant, which makes it possible for the Searcher to prioritize the other boxes later on in a search sequence to still achieve optimality.
The several cases in which an optimal pure strategy exists for the Searcher we find in this paper all meet this general observation.

A natural question to ask in general is whether we can determine the best search sequence that minimizes the expected total search time regardless of where the Hider hides.
In addition, what is the gap between an optimal mixed search strategy and the best search sequence when the latter is not optimal?
In Section~\ref{sec:best-2boxes} we make a conjecture of the best search sequence for a very special case but it appears rather difficult to determine the best search sequence in general.

\section*{Acknowledgements} This material is based upon work supported by the National Science Foundation under Grant No. CMMI-1935826.

\end{document}